\documentclass[11pt, oneside]{article}
\usepackage{titlesec}
\usepackage{placeins}
\usepackage{xcolor}
\usepackage{subcaption}
\usepackage[affil-it]{authblk}
\usepackage{float}
\usepackage[hyphens]{url}
\usepackage{latexrc}
\usepackage{bold-extra}
\usepackage{setspace}
\usepackage{etoolbox}
\usepackage{scrextend}
\patchcmd{\abstract}
  {\bfseries\abstractname}
  {\scshape\bfseries\abstractname}
  {}{}

\mc A
\mc B
\newcommand{\D}{\mathfrak{D}}
\newcommand{\Inc}{\mathrm{Inc}}
\newcommand{\Exc}{\mathrm{Exc}}
\let\oldsubstack\substack
\renewcommand{\substack}[1]{\mathclap{\oldsubstack{#1}}}

\title{\large\bfseries\scshape The Inherent Instability of Disordered Systems}
\author{\small Taeer Bar-Yam}
\author{\small Owen Lynch}
\author{\small Yaneer Bar-Yam}
\affil{\vspace{-2mm}\footnotesize New England Complex Systems Institute\\
277 Broadway Cambridge, MA 02139}
\date{\vspace{-0.5cm}\small \today}

\begin{document}
\maketitle

\begin{abstract}
The Multiscale Law of Requisite Variety is a scientific law relating, at each scale, the variation in an environment to the variation in internal state that is necessary for effective response by a system. While this law has been used to describe the effectiveness of systems in self-regulation, the consequences for failure have not been formalized. Here we use this law to consider the internal dynamics of an unstructured system, and its response to a structured environment. We find that, due to its inability to respond, a completely unstructured system is inherently unstable to the formation of structure. And in general, any system without structure above a certain scale is unable to withstand structure arising above that scale. To describe complicated internal dynamics, we develop a characterization of multiscale changes in a system. This characterization is motivated by Shannon information theoretic ideas of noise, but considers structured information. We then relate our findings to political anarchism showing that society requires some organizing processes, even if there is no traditional government or hierarchies. We also formulate our findings as an inverse second law of thermodynamics; while closed systems collapse into disorder, systems open to a structured environment spontaneously generate order.
\end{abstract}

\pagebreak
\doublespacing

In this paper we explore the dynamics of disordered systems. The traditional approach to studying systems, which considers macroscopic forces and microscopic thermodynamics, has been extended by information theoretic models which describe systems by the scales at which they are organized~\cite{cprof}. Prior works have also used information theoretic models to characterize whether a system functions successfully. However, they have only considered whether systems succeed or fail, not the implications of that failure~\cite{ashby,msvariety}. They also only consider systems functioning completely according to specification. By extending those frameworks to consider deviations from expected behavior we can understand the structural changes that result from failures in certain contexts.

Our objective is to characterize how the structure of a system changes spontaneously and/or as it interacts with an environment. We find that in the context of an ordered environment order arises in previously disordered systems. We use multiscale information theory to define the structure of a system and Ashby's Law of Requisite Variety~\cite{ashby} to analyze the failure conditions that give rise to change. Where internal self-organizing structures occur, we treat them as part of the external environment, enabling the application of Ashby's law to spontaneous internal change. For this purpose we will consider any change to the structure of the system to be a failure on the part of the system to maintain its state.

We start in Section \ref{simple} by considering systems consisting of coordinated blocks of parts, and then in Section \ref{infotheory} we expand our analysis to systems with much more general structure, defined by the behavior of different parts of the system sharing information. In section \ref{conclusion} we present our main conclusion that order arises in disordered systems, and relate it to the second law of thermodynamics, which describes the increase in disorder of a closed system. We also describe implications for social systems and the philosophy of anarchism.

\section{\scshape Emergence of Order in Simple Systems}
\label{simple}

Given a system in an environment, where the environment can be in each of $v$ possible states that each requires a distinct response by the system in order for it to ``succeed'', the system tautologically needs $v$ responses available to it in order to guarantee success. This mathematical theorem is the Law of Requisite Variety stated by Ashby~\cite{ashby}. A multiscale version of this law has been stated~\cite{msvariety}, which adds the assumption that the environmental behaviors occur at different scales, and that the system must respond at the appropriate scale. In a simplified model, the system is assumed to be composed of $N$ atomic parts (these could be atoms, molecules, cells, people, etc.), arranged into subsystems that have perfect internal coordination. The number of parts in a given subsystem determines the scale of action it can take. An environmental change at scale $k$ by definition requires an action performed by $k$ parts in coordination, and thus only subsystems of size $k$ or greater will be capable of responding.

For example constructing a building requires the coordinated effort of many people. Furthermore, if buildings on sand, dirt, and rocks each require a different method of construction, then not only is a coordinated group needed, but one which can perform three different types of construction.

The development of the Multiscale Law of Requisite Variety leads to the notion of a trade-off between coordination and flexibility. A system can have many parts coordinated to allow large scale responses, or it can have those parts independent to allow independent, and varied responses~\cite{msvariety}. For example, if a person has the capacity to know how to construct up to ten different things, then each can learn ten ways of making a small hut by themselves, so that there are fifty different small buildings that can be built, or they can coordinate to learn to construct ten types of larger buildings. We can formalize this trade-off by defining $v$ as the variety of an individual part (which for simplicity we assume to be the same for all parts), and $n(k)$ as the number of different $k$-member fully coordinated subsystems. Then $D(k) = vn(k)$ is the variety in response at a particular scale for the whole system, and $C(k) = \sum_{k\'=k}^N D(k\')$ is the variety in responses at scale $k$ or greater, the ``complexity'' at scale $k$ (See Figure~\ref{comp}). The key result then, is that when we sum over $C(k)$ at all scales, we obtain the total variety across all parts of the system treated independently,
\begin{equation}
  \sum_{k=1}^N C(k) = Nv.
\end{equation}
This quantity depends only on the number of parts and the variety of those parts, not their organization. It is therefore fixed with respect to reorganizations of the system. Any reorganization that increases the complexity at one scale must simultaneously decrease the complexity at some other scales. As parts become linked to each other in behavior, the scale increases, and the variety at the individual scale across the system decreases. The correct balance of this trade-off depends on the environment the system must respond to: if the environment can take on many small scale states that must be responded to differently, an effective system will have many independent parts that can respond independently. If the environment poses large scale threats, an effective system will have global coordination.
\renewcommand{\thesubfigure}{\alph{subfigure}}
\begin{figure}
  \begin{subfigure}[b]{0.5\textwidth}
    \renewcommand{\captionfont}{\Large}
    \captionsetup{justification=justified,singlelinecheck=false}
    \caption{}
    \includegraphics{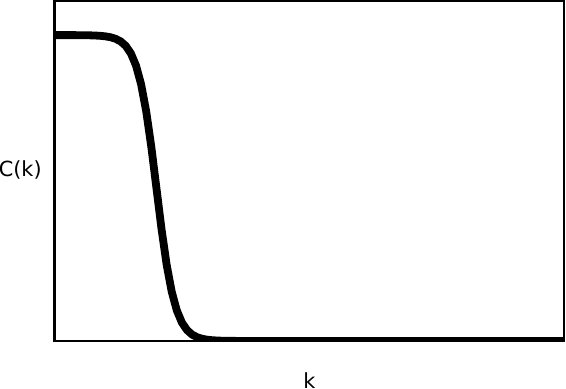}
    \label{compgas}
  \end{subfigure}
  \hspace{0.5cm}
  \begin{subfigure}[b]{0.5\textwidth}
    \renewcommand{\captionfont}{\Large}
    \captionsetup{justification=justified,singlelinecheck=false}
    \caption{}
    \includegraphics{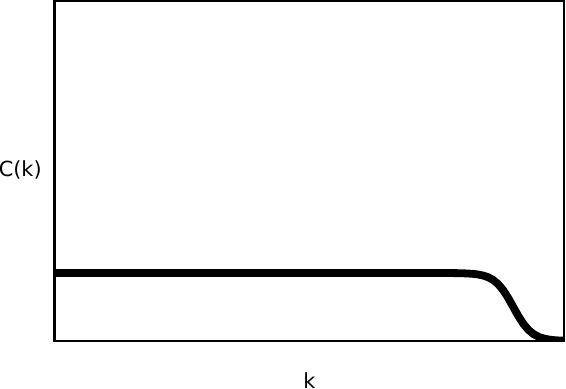}
    \label{comprb}
  \end{subfigure}

  \caption{Example Complexity Profiles. (\protect\subref{compgas}) Shows the complexity profile of a gas or other disordered system reflecting that at a fine scale every particle has an independent position and momentum which must be specified to describe the state of the system. At larger scales the different states are indistinguishable from each other, and so the variety at larger scales is zero. (\protect\subref{comprb}) Shows the complexity profile of a rigid body. A rigid body is relatively easy to describe, since all of the particles have the same momentum and the same position relative to each other, so only one must be specified. However, the position, movement, and orientation of the rigid body is information carried by all of the parts, and so is distinguishable at a large scale.}
  \label{comp}
\end{figure}

This formalism has been used~\cite{msvariety} to show that although central control and hierarchy create large-scale coordination, they limit the complexity of a system at smaller scales (See Figure~\ref{comprb}), which makes the system unable to function effectively in environments with complexity at many scales, referred to as ``complex'' environments. Here we consider systems without any coordination, where each part of the system acts entirely independently.

A first observation we can make is that a disordered system is unfit to respond to any variation in the environment at a scale larger than the individual parts. For example, if a social system of this type were faced with a large drought, it would be unable to bring in large quantities of water from elsewhere, as that would require a coordinated effort to e.g. build pipes. Beyond being unsuited to certain environments, though, a disordered system is inherently unstable to the emergence of order. To see this, we must extend our discussion to say something about the internal dynamics of systems.

When we talk about a system having a structure, we are talking about a set, or ensemble, of states the system is allowed to be in, possibly augmented by probabilities of the system being in each state. The system state changes within this ensemble, but given a static structure the ensemble itself does not change. One can also describe the dynamics of the structure as the ensemble changes. Consider a box with a partition separating two gasses. If one removes the partition, the ensemble of allowed states changes, and therefore the complexity profile also changes. Structural change is the \textit{breakdown} of expected behavior. When someone gets cancer, we consider this a departure from the expected functioning of a human body. In a society, revolution is the birth of a new social order, and the death of an old one. This begins with part of the system acting in a way that violated the defined behavior of the original social order. Even the construction of a new building can be considered to be the breakdown of the system that didn't have a building there. When a system stops functioning `properly', the complexity profile changes to reflect the new behavior of the system. Generally, parts of the system might partially deviate in their behavior, such as a cell that continues to function but over-replicates, or a citizen who goes to work but at night vandalizes buildings. We will consider this general case in more detail later. For now, we make the simplifying assumption there is a well defined subset of the system's parts that has taken on entirely new behavior.

The examples of cancer and revolution can be considered in the same way as external threats, like an infection or invasion respectively. There is on the one hand the system, and on the other hand the problem which the system must solve, even if the problem originates within the old system. Mathematically: when a subset of the system takes on a new behavior, that subsystem becomes considered part of the environment that the system must now deal with as an external threat. In order to maintain itself in the face of such a disturbance, the system must be able to either re-program those parts to their appropriate behavior and thus reabsorb them, or destroy/eject and replace them. This occurs in biological systems with processes such as autophagy or apoptosis, in which misbehaving cells are either consumed by other cells or forced to self-destruct using chemical signaling, and in social systems with reeducation, ostracization, imprisonment, and capital punishment.

When we consider a system defined by its lack of coordination (See Figure~\ref{compgas}), the emergence of order at a large scale $\kappa$ is just such a violation of defined behavior. As per our model, we then consider this a part of the environment, and in particular, an environmental event occurring at scale $\kappa$. By the multiscale law of requisite variety, the system is unable to respond effectively to such an event, so it can neither reabsorb the malfunctioning component nor ensure that it is properly ejected. The system change cannot be reversed. 

We can easily extend this reasoning to systems with structure at a scale larger than the individual using the same reasoning as above. We deduce that as long as a system lacks structure above a certain scale, it is not robust to structure arising above that scale.

There are several assumptions we have made in coming to this conclusion that should be reiterated. First, we have assumed that deviation occurs in the functioning of the system. While all real systems deviate to some extent, over short periods of time systems may be best approximated as acting flawlessly. Second, considering the emergence of large scale structure, we assume that coordination is possible. In social and biological systems, this is apparent, but not at every scale. For instance until recently global communication among people has been quite limited. In certain physical systems such as a noble gas, there can be almost no interaction between particles. Without the ability to interact and coordinate, no larger scale structure (aside from fluctuations) can arise, and so the system is stable. Finally, an environment that has extremely high variety at a small scale will not allow large scale structure to persist. In such a case, any structure that develops at the large scale necessarily sacrifices variety to do so, and therefore leaves itself unable to handle the variety of the environment. In this case the large scale structure would be unable to maintain itself, and would collapse back into unstructured parts. The system would then in some sense be maintaining itself. However, if large scale structure \textit{can} arise, the unstructured system itself will be unable to respond effectively.

We also assumed for simplicity that there is a well defined subsystem that misbehaves completely. We can extend our model to the general case in which there is no clearly defined subsystem that has been co-opted, but the system is not behaving according to its defined behavior. As a motivating example, consider a change in the priorities of the population that causes an arbitrarily small change in the standard work week from $40$ hours to $40-\ep$. For the sake of this example we suppose that this policy causes people to spend less time at work, and more time at home with their families or out with their friends. The structure of the system has changed; relationships have become weaker or stronger, and people will therefore be coordinated in different ways. However, if we imagine trying to reassert the old system and return the coordination to what it was before, there is no smaller part of the system to reprogram or eject; every part of the system is affected. In our current model, we must say that the entire system is different. On the other hand, the change can be made arbitrarily small. People are acting essentially the same as they were before. The question, then, is what variety and scale of responses is needed to effectively respond to the \textit{change} in the system rather than to the whole new system.

The need to quantify the change in the system makes clear why the assumption of completely deviating subsystems simplified the analysis previously. When we have completely deviating subsystems, the variety in the new system and the variety in the change are equivalent. Once we formalize the more general case, we will prove this assertion directly.

\section{\scshape Information Theoretic Formalism for Structural Change}
\label{infotheory}
To answer the new question---what is the variety and scale in the change of the system---we introduce an information theoretic formalism. Until now we have defined our system as composed of distinct subsystems that are fully coordinated internally and uncoordinated between subsystems. Now we consider a general system composed of a set of atomic parts $A$ that can each be in a set of states $S_a (a\in A)$. Dependencies between these parts can be seen in the information that is redundant between them (the shared information).
We summarize previous work formally defining the complexity profile on these systems~\cite{cprof}. In some places we use new notation that simplifies the definitions. Given atomic parts $A=\set{a_1,\ldots,a_n}$, we consider formal expressions of the form $a_{i_1};\cdots;a_{i_k}|a_{j_1},\ldots,a_{j_l}$ (where two are the same if they differ only by the ordering of elements separated by `$;$' or `$,$'). These are \textit{dependencies} and they represent regions in the information Venn diagram (See Figure~\ref{H(a)}). One can think of `$;$' as union, `$,$' as intersection, and `$|$' as set difference.

In this paper we let $\Inc(x)$ for a dependency $x$ denote the set of ``included'' atomic parts (those before the `$|$') and $\Exc(x)$ the set of ``excluded'' atomic parts (those after the `$|$'). It is worth noting that dependencies are completely defined by their included and excluded parts.

If all atomic parts are used exactly once in a dependency, it is an \textit{irreducible dependency} and it represents one of the minimal regions in the Venn diagram (See Figure~\ref{H(a)}). Other dependencies are reducible and represent larger regions of the ven diagram, unions of irreducible dependencies. Formally a dependency $x$ is irreducible if $\Inc(x)\sqcup\Exc(x)=A$, and reducible otherwise. Let $\D_\A$ be the set of irreducible dependencies.

The intuition of dependencies as regions in the Venn diagram appears in the formalism when we define a function $I_\A$ on dependencies that gives us the information in that region. We assume that we have an entropy function $H$ on sets of atomic parts; one such function is the Shannon entropy. We define $I_\A$ by a system of equations, one for each set of atomic parts $B\ins A$:
  \begin{equation} \label{eq:defineI}
    \sum_{\substack{x\in\D_\A\\\Inc(x)\cap B\ne\0}}I_\A(x) = H(B).
  \end{equation}
  There are $2^{\abs{A}}$ independent equations in $2^{\abs{A}}$ unknowns, which gives us a unique solution for all of the irreducible $I_\A(x)$. An example of one of these equations, $I_\A(a|b,c) + I_\A(a;b|c) + I_\A(a;c|b) + I_\A(a;b;c) = H(a)$, is depicted in Figure~\ref{H(a)}. 
 
  Reducible dependencies define larger regions in the Venn diagram. We wish to define $I_\A(x)$ as the sum over the information in the irreducible dependencies whose regions are contained in the region given by $x$. In this paper, to capture the idea of one dependency being contained in another we define a partial ordering on dependencies where $x\le y$ when the region defined by the dependency $x$ is contained in the region defined by the dependency $y$. Formally, we have $x\le y$ when $x$ includes and excludes all of the parts $y$ does (and possibly more), or in equations
  \begin{equation}
    x\le y \iff \Inc(x)\supseteq\Inc(y)\wedge\Exc(x)\supseteq\Exc(y).
  \end{equation}
  Now we complete the definition of $I(x)$. For some (not necessarily irreducible) dependency $x$,
  \begin{equation}
    I_\A(x)=\sum_{\substack{z\in\D_\A\\z\le x}}I_\A(z).
  \end{equation}

  If two systems $\A$ and $\B$ contain some parts in common ($A\cap B\ne\0$), and a dependency $x$ contains only parts in both ($\Inc(x)\cup\Exc(x)\ins A\cap B$) then $x$ is a dependency in both systems. We prove in Appendix~\ref{Iwelldefined} that the information in $x$ is independent of which we calculate it in,
  \begin{equation}
    I_\A(x)=I_\B(x).
  \end{equation}
  Since the subscript is irrelevant, hereafter we write $I(x)$.

\def\firstcircle{(0,0) circle (1.5cm)}
\def\secondcircle{(60:2cm) circle (1.5cm)}
\def\thirdcircle{(0:2cm) circle (1.5cm)}
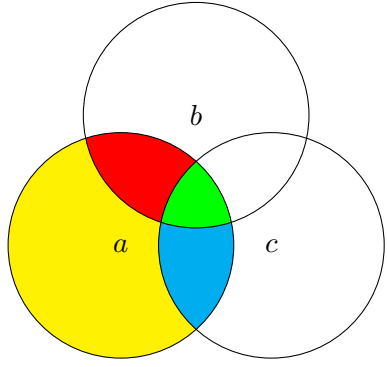
\begin{figure}
  \centering
  \begin{tikzpicture}
    \begin{scope}
      \clip \secondcircle;
      \fill[red] \firstcircle;
    \end{scope}

    \begin{scope}
      \clip \thirdcircle;
      \fill[cyan] \firstcircle;
    \end{scope}

    \begin{scope}
      \clip \thirdcircle;
      \clip \secondcircle;
      \fill[green] \firstcircle;
    \end{scope}
    \begin{scope}
      \begin{scope}[even odd rule]
        \clip \thirdcircle (-3,-3) rectangle (3,3);
        \clip \secondcircle (-3,-3) rectangle (3,3);
        \fill[yellow] \firstcircle;
      \end{scope}
      \draw \firstcircle node {$a$};
      \draw \secondcircle node {$b$};
      \draw \thirdcircle node {$c$};
    \end{scope}
  \end{tikzpicture}
  \caption{A Venn diagram with some of the irreducible dependencies colored in. Green is $(a;b;c)$, Red is $(a;b|c)$, Blue is $(a;c|b)$, Yellow is $(a|b,c)$. When we sum over $I(r)$ for all of the colored regions/dependencies $r$, we get $H(a)$, the total entropy in the variable $a$. Other dependencies can all be seen as unions of irreducible dependencies. For instance, the yellow and blue regions combined are $(a|b)$, blue and green together are $(a;c)$, etc.}
  \label{H(a)}
\end{figure}

We assume for simplicity that all parts have the same scale and we can define $s(x)$, the scale of a dependency $x$, as $s(x) = \abs{\Inc(x)}$, the number of components who share information in this dependency. Then $C_\A(k)$, the complexity of $\A$ at scale $k$, is defined by
\begin{equation}
  C_\A(k) = \sum_{\substack{x\in\D_\A\\ s(x)\ge k}}I(x).
\end{equation}
A visualization of $C(k)$ for a three part system is pictured in Figure~\ref{C(k)}.

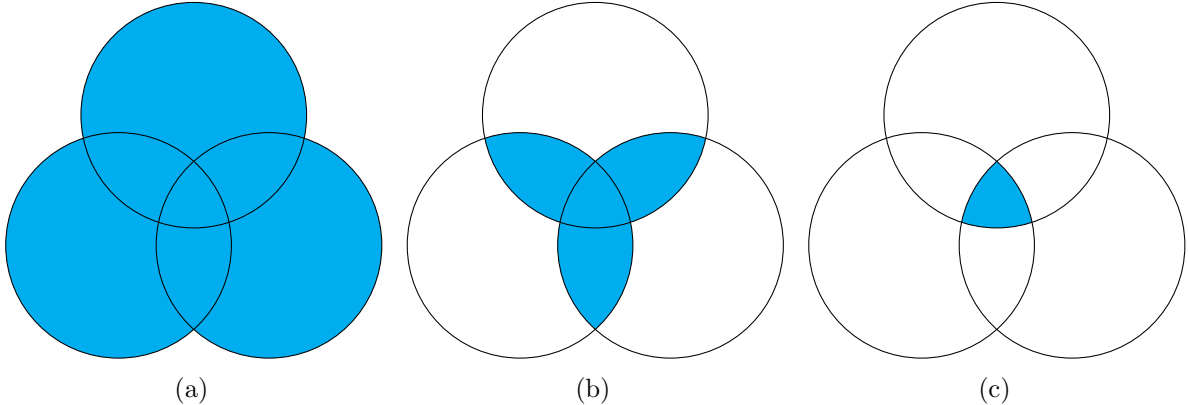
\begin{figure}
  \centering
  \begin{subfigure}[b]{0.3\textwidth}
    \begin{tikzpicture}
      \begin{scope}
        \fill[cyan] \firstcircle;
        \fill[cyan] \secondcircle;
        \fill[cyan] \thirdcircle;
      \end{scope}
      \draw \firstcircle node {};
      \draw \secondcircle node {};
      \draw \thirdcircle node {};
    \end{tikzpicture}
    \caption{}
    \label{C(1)}
  \end{subfigure}
~
  \begin{subfigure}[b]{0.3\textwidth}
    \begin{tikzpicture}
      \begin{scope}
        \clip \thirdcircle;
        \fill[cyan] \firstcircle;
      \end{scope}
      \begin{scope}
        \clip \secondcircle;
        \fill[cyan] \firstcircle;
      \end{scope}
      \begin{scope}
        \clip \secondcircle;
        \fill[cyan] \thirdcircle;
      \end{scope}
      \draw \firstcircle node {};
      \draw \secondcircle node {};
      \draw \thirdcircle node {};
    \end{tikzpicture}
    \caption{}
    \label{C(2)}
  \end{subfigure}
~
  \begin{subfigure}[b]{0.3\textwidth}
    \begin{tikzpicture}
      \begin{scope}
        \clip \thirdcircle;
        \clip \secondcircle;
        \fill[cyan] \firstcircle;
      \end{scope}
      \draw \firstcircle node {};
      \draw \secondcircle node {};
      \draw \thirdcircle node {};
    \end{tikzpicture}
    \caption{}
    \label{C(3)}
  \end{subfigure}
  \caption{A visual representation of the information contributing to $C(k)$. (\protect\subref{C(1)}) $C(1)=H(\A)$ contains all the information in $\A$, (\protect\subref{C(2)}) $C(2)$ has all the information shared among two or more atomic parts, and (\protect\subref{C(3)}) $C(3)$ has only the information shared among all three.}
  \label{C(k)}
\end{figure}

We now develop a new formalism for system change motivated by the treatment of noise in Shannon information theory~\cite{shannon}. Noise changes a signal so that it is not received as it was intended. Similarly, we would like to describe a system that is not behaving as it is intended, and we are describing that system by the information in its behavior. In his theory of communication, Shannon formalizes noise by considering the intended signal $x$ and the received signal $x\'$ together. He looks at the entropy of one given that the other is known: the conditional entropy $H(x\'|x)$ (which he writes $H_x(x\')$). This is the amount of new information in the noisy signal, that was not present without noise. We will do something similar in considering the multiscale noise associated with a distorted system.

We have an ideal system $\A$ and the ``noisy'' version $\A\'$. In order to study the inter-system conditional entropies, we will consider both systems as part of an expanded system $\A\cup\A\'$. As in communication theory, we are interested in the new information. Unlike communication, we are interested in the new information in each dependency; if there was information present in the system that now appears at a larger scale, this is relevant to the change in structure of the system. We therefore define $I(x|y)$ for \textit{dependencies} $x$ and $y$ (the new information in the dependency). We want $I(x|y)$ to be the information in the dependency $x$, given we know the information in the dependency $y$. On the Venn diagram, if $x$ and $y$ represent regions $R$ and $S$, then $(x|y)$ represents the part of the region $R$ that is not part of $S$. $I(x|y)$, then, is the sum over the information in irreducible dependencies representing regions contained in $R$ but not in $S$:
\begin{equation}
  I(x|y) = \sum_{\substack{z\in\D_{\A\cup\A\'}\\
      z\le x\\
      z\not\le y
  }} I(z) 
\end{equation}
For our purposes, this level of generality suffices. For a definition of $I$ on a generalized dependency representing any region of a system's Venn diagram see Appendix~\ref{generalI}.

We can now define the complexity as a function of scale for the noise in the system,
\begin{equation}
  C_{\A\'|\A}(k) = \sum_{\substack{x\in\D_\A\\ s(x)\ge k}}I(x\'|x)
\end{equation}
where $x\'$ is the equivalent of $x$ in the primed system (See Figure~\ref{noise}).
\newcommand{\cylinder}[5]{
  #5
    \fill[top color=#3!50!#4,bottom color=#3!10,middle color=#3,shading=axis,opacity=0.15] (0+#1,0+#2) circle (2cm and 0.5cm);
  \fi
  \fill[left color=#3!50!#4,right color=#3!50!#4,middle color=#3!50,shading=axis,opacity=0.15] (2+#1,0+#2) -- (2+#1,4+#2) arc (360:180:2cm and 0.5cm) -- (-2+#1,0+#2) arc (180:360:2cm and 0.5cm);
  \fill[top color=#3!90!,bottom color=#3!2,middle color=#3!30,shading=axis,opacity=0.15] (0+#1,4+#2) circle (2cm and 0.5cm);
  \draw (-2+#1,4+#2) -- (-2+#1,0+#2) arc (180:360:2cm and 0.5cm) -- (2+#1,4+#2) ++ (-2,0) circle (2cm and 0.5cm);
  \draw[densely dashed] (-2+#1,0+#2) arc (180:0:2cm and 0.5cm);
}
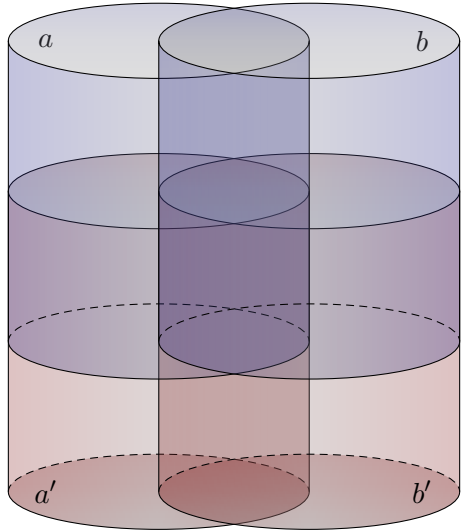
\begin{figure}
  \centering
  \begin{tikzpicture}
    \definecolor{redlight}{HTML}{FF8888}
    \definecolor{reddark}{HTML}{FF0000}
    \definecolor{bluelight}{HTML}{8888FF}
    \definecolor{bluedark}{HTML}{0000FF}
    \cylinder{0}{0}{redlight}{reddark}{\iftrue}
    \node at (-1.5,6) {$a$};
    \cylinder{2}{0}{redlight}{reddark}{\iftrue}
    \node at (3.5,6) {$b$};
    \cylinder{0}{2}{bluelight}{bluedark}{\iffalse}
    \node at (-1.5,0) {$a^\prime$};
    \cylinder{2}{2}{bluelight}{bluedark}{\iffalse}
    \node at (3.5,0) {$b^\prime$};
  \end{tikzpicture}

  \caption{$C_{\A^\prime|\A}(k)$ can be visualized by taking the Venn diagram for the two systems $\A$ (the information describing $\A$ is represented in the blue region) and $\A^\prime$ (the information describing $\A^\prime$ is represented in the red region), placing them on top of one another and extruding them into each other (the resulting overlap, which has the information shared between $\A$ and $\A^\prime$, is both blue and red making it appear purple). The regions now represent some of the information in the joint system $\A\cup\A^\prime$. The region at the bottom colored only in red is the information counted in $C_{\A^\prime|\A}$, the new information in the dependencies of $\A^\prime$ that wasn't in those dependencies in $\A$. This is not the complete Venn diagram for the expanded system $\A\cup\A^\prime$, but it is visually intuitive and has the regions that are necessary to define $C_{\A^\prime|\A}(k)$.}
  \label{noise}
\end{figure}

We note that the existence of negative complexities has been proven, and presents difficulties for interpretation within Ashby's law. For instance, it is unclear why negative complexity can cancel with positive complexity when systems are combined. For our purposes, we will assume $I(x)$ is always non-negative.

We note a few simple cases that test the reasonableness of our definition for $C_{\A\'|\A}(k)$ and shows that it agrees with our original analysis (See Section~\ref{simple}) in the simple case. Proofs for these statements can be found in Appendix~\ref{sanitycheck}.

\begin{enumerate}
  \item If there is no noise, the noise has complexity $0$ at all scales.
  \item \label{item:newsys} Given an entirely new system $\B$ having no relation to $\A$, the complexity of the noise $C_{\B|\A}$ will just be the complexity of $\B$. 
  \item \label{item:partsys} If $\A$ and $\B$ are independent subsystems composing a larger system $\A\cup\B$ then the complexity of noise affecting the system as a whole is the sum of the noise affecting the subsystems. In particular, if one of the subsystems is unaffected then the complexity of the noise of the whole system equals that of the affected subsystem.
  \item Combining points \eqref{item:newsys} and \eqref{item:partsys} we can recover our prior result of Section \ref{simple} that completely deviating subsystems can be treated as the subsystem becoming part of the environment. In this case, the complexity of the change is the complexity of the altered subsystem.
\end{enumerate}

Our motivation for developing this framework has been to apply the Law of Requisite Variety to the noise complexity profile. We prove in Appendix~\ref{apx:noisehascx} that if the unaltered system $\A$ has no complexity above some scale $\kappa$, and in the altered system $A\'$ there is complexity at scale $\kappa$, then the change $C_{\A'|A}(k)$ has complexity at scale $\kappa$. Applying the law of requisite variety we infer that the system is unequipped to respond to the noise.

\section{\scshape Conclusions}
\label{conclusion}

One implication of the framework we have developed is an inverted second law of thermodynamics. The second law of thermodynamics states that in a closed system entropy, which is equal to scale $1$ complexity~\cite{cprof}, increases. Our analysis states that an open system exposed to a structured environment will develop complexity at all scales. By spontaneous variation larger scale complexity will arise, and by the multiscale law of requisite variety the system will adopt that large scale complexity; over enough time all scales will be reached. The two laws and the parallel between them is depicted in figure~\ref{fig:secondlaw}.
\begin{figure}
  \centering
  \includegraphics{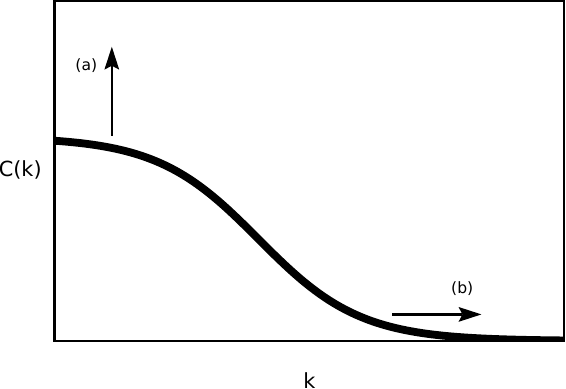}
  \caption{A complexity profile showing the effects of $(a)$ the second law of thermodynamics, which says that in a closed system $C(0)$ increases, and $(b)$ the results of our analysis, which implies that in a system exposed to a structured environment the maximum scale at which the system has complexity increases. As can be seen in the figure, these laws are inverse in that they are reflected across the $x=y$ axis.}
  \label{fig:secondlaw}
\end{figure}

With our general framework, we can investigate a range of topics, including changes in societies. As an example we can explore the implications to a conception of anarchy as a society of completely autonomous individuals. Our analysis shows that such anarchy is unstable to the formation of structure. In Section \ref{simple} we only considered the possibility of a fragment of the society deviating completely; now we can also talk about order emerging gradually within the system. For instance, a trade network could start with neighbors bartering for some of their needs, but then grow to encompass an entire region and become used for more goods, as differences in local production capabilities make long-distance trade desirable, and essential to conduct. By Theorem \ref{noisehascx}, the change that created this trade network has complexity at a scale greater than the individual. As a consequence of the multiscale law of requisite variety, in order for the system to respond to the change it must have complexity at a scale greater than the individual, which it does not. Thus, the system has no means of regulating the emergence of a large-scale structure. Since a large-scale structure is a system breakdown with regard to the definition of an unstructured society, the original system cannot survive---the change will persist.

This concept of anarchy is often presented as a solution to the problems of central control in society~\cite{anarchistfaq}. Our results imply that without an alternative form of large scale coordination, anarchy will ultimately be unsuccessful. Modern anarchists do not assume that there will be no large scale coordination~\cite{smallscaleanarchy}. For instance, anarcho-syndicalists propose organization into worker-groups, and anarcho-capitalists believe in markets. The framework of this paper does not allow us to identify what is a viable replacement for central control. However, developing a theory of which systems are stable and which systems are not stable is an important step towards answering this question.

One of the original works on the failure of anarchy was Leviathan, by Thomas Hobbes, which famously characterized life in such a society as ``nasty, brutish, and short''. Hobbes advocated monarchy as an alternative. Perhaps Hobbes' absolutist view was incubated by his earlier study of physics, which at the time only dealt with rigid bodies and easily solvable equations. In this modern age, we too look to science to explain society. However, rather than relying on analogies to physical systems, we have the ability to study society directly with information theoretic laws and draw more sophisticated conclusions. The multiscale law of requisite variety implies that there is no single organizational ``silver bullet''---different organizational strategies are best suited to different environments. Our addition of a formalism for multiscale noise opens the door to characterizing the stability and dynamics of complex systems. A more general approach may help us discover how to change our society from the inside to better fit the complexity profile of modern challenges.

\section*{Acknowledgements}

We would like to thank Ben Allen and Seth Frey for helpful comments.

\pagebreak
\appendix
\titleformat{\section}[display]
  {\pagebreak\normalfont\bfseries\filcenter}{\scshape \bfseries \large Appendix \thesection}{0ex}
  {\vspace{-1cm}}
\section{}
\label{Iwelldefined}
In this Appendix, we prove a lemma showing that $I_\A(x)$ is independent of the system in which it is calculated, and thus that $I(x)$ is well defined without the subscript. $I_\A(x)$ is defined in the manuscript in equation \eqref{eq:defineI}: For each $U\ins A$,
\begin{equation}
  \sum_{\substack{x\in\D_\A\\\Inc(x)\cap U\ne\0}}I_\A(x)=H(U).
\end{equation}
Note that $H(U)$ does not require a subscript, since entropy is defined on the variables, independent of the enclosing system.
\begin{lemma}
  \label{lemma:Iwelldefined}
  If we have systems $\A$ and $\B$ with overlapping parts, in other words $A\cap B\ne\0$, and we have a dependency $x$ with $\Inc(x)\cup\Exc(x)\ins A\cap B$ so that it can be considered a dependency of either system, then $I_\A(x)=I_\B(x)$. Note that these are exactly the conditions needed for $I_\A(x)$ and $I_\B(x)$ to both be defined. Thus any time $I_\A(x)$ and $I_\B(x)$ are both defined, they are equal.
\end{lemma}
\begin{proof}
  We will prove that $I_\A(x)=I_{\A\cup\B}(x)$, and since there is nothing distinguishing $\A$ and $\B$, we will also have $I_\B(x)=I_{\A\cup\B}(x)$ proving the lemma by transitivity. 

  We start by proving it only for irreducible dependencies of $\A$. The equations that define $I$ for an irreducible dependency $x$ of $\A$ are for each $U$ subset of $A$
  \begin{equation}
    \sum_{\substack{x\in\D_\A\\\Inc(x)\cap U\ne\0}}I_\A(x)=H(U)
  \end{equation}
  and caluclating $I(x)$ in $\A\cup\B$ we have 
  \begin{equation}
    I_{\A\cup\B}(x) = \sum_{\substack{y\in\D_{\A\cup\B}\\y\le x}}I_{\A\cup\B}(y).
  \end{equation}
  We want to show, then, that $I_{\A\cup\B}(x)$ satisfies the definition of $I_\A(x)$ for $x\in\D_\A$. Since the definition is a system of linear equations with the same number of equations as unknowns, there is a unique solutions and so we will have that $I_{\A\cup\B}(x) = I_\A(x)$. For all $U\ins A$, we want
  \begin{subequations}
    \begin{align}
      \sum_{\substack{x\in\D_\A\\\Inc(x)\cap U\ne\0}}I_{\A\cup\B}(x) &= H(U)\\
      \sum_{\oldsubstack{x\in\D_\A\\\Inc(x)\cap U\ne\0}}\sum_{\oldsubstack{y\in\D_{\A\cup\B}\\y\le x}}I_{\A\cup\B}(y) &= H(U).\\
      \intertext{In order to collapse this into one sum, we need to make sure there is no double counting, that we never reach the same term in the inner sum while on different iterations of the outer sum. In other words, that the outer sum is over disjoint sums. So we must show given $x\ne x\'\in\D_\A$ and $y,y\'\in\D_{\A\cup\B}$ with $y\le x$ and $y\'\le x\'$, we must show $y\ne y\'$. But since $x,x\'$ are irreducible in $A$ and $x\ne x\'$, $\E a\in A.\; a\in \Inc(x)\cap \Exc(x\')$. Then $y\le x\thus a\in\Inc(y)$ and likewise $a\in\Exc(y\')$, and so $y\ne y\'$. So the sum is over disjoint sums, and we can combine it into one sum:}
      \sum_{\substack{x\in\D_\A\\\Inc(x)\cap U\ne\0\\y\in\D_{\A\cup\B}\\y\le x}}I_{\A\cup\B}(y) &= H(U).\\
      \intertext{If $\Inc(x)\cap U\ne\0$ and $y\le x$, meaning $\Inc(y)\supseteq\Inc(x)$, then $\Inc(y)\cap U\ne\0$. Additionally, for every $y\in\D_{\A\cup\B}$ with $\Inc(y)\cap U\ne\0$, we can let $x$ be $y$ with all of the elements of $B$ removed, and $y\le x$. We also know since $U\ins A$ that $\Inc(x)\cap U\ne\0$ (the overlap is not lost when we restrict to $A$). All of this means we can rewrite as}
      \sum_{\substack{y\in\D_{\A\cup\B}\\\Inc(y)\cap U\ne\0}}I_{\A\cup\B}(y) &= H(U).
    \end{align}
  \end{subequations}
  But this is a subset of the equations defining $I_{\A\cup\B}(y)$.\\

  Now for reducible dependencies of $\A$, we have 
  \begin{subequations}
    \begin{align}
      I_\A(x) &= \sum_{\substack{y\in\D_\A\\y\le x}}I_\A(y)\\
              &= \sum_{\substack{y\in\D_\A\\y\le x}}I_{\A\cup\B}(y)\\
              &= \sum_{\oldsubstack{y\in\D_\A\\y\le x}}\sum_{\oldsubstack{z\in\D_{\A\cup\B}\\z\le y}}I_{\A\cup\B}(y).\\
      \intertext{Just as before, there is no double counting and we can collapse the double sum:}
      &= \sum_{\substack{y\in\D_\A\\y\le x\\z\in\D_{\A\cup\B}\\z\le y}}I_{\A\cup\B}(y).\\
      \intertext{For every $z\in\D_\A$, we can let $y$ be $z$ with all of the elements of $A\setminus B$ removed, and $z\le y\le x$. This means the sum sums over every dependency in $\D_\A$.}
      &= \sum_{\substack{z\in\D_{\A\cup\B}\\z\le x}}I_{\A\cup\B}(y)\\
      &= I_{\A\cup\B}(x)
    \end{align}
  \end{subequations}
\end{proof}

\section{}
\label{generalI}
For a system $\A$, an arbitrary region of the information Venn diagram can be described by some combination of unions (written `$,$'), intersections (written `$;$'), set differences (written `$|$'), and elementary regions denoted by elements $a\in A$. Expressions of this form we call generalized dependencies. We define $I(\xi)$ for some generalized dependency $\xi$ to be the sum of $I(x)$ for all irreducible dependencies $x$ that represent regions contained in the region represented by $\xi$. To complete the definition we must specify which dependencies represent such regions. We can recursively define a truth function $R_\xi(x)$ that tells us if the region defined by $x$ is in the region defined by $\xi$: 
  \begin{subequations}
    \begin{align}
      R_a(x) &= a\in\Inc(x)\\
      R_{(\xi,\zeta)}(x) &= R_\xi(x) \vee R_\zeta(x)\\
      R_{(\xi;\zeta)}(x) &= R_\xi(x) \wedge R_\zeta(x)\\
      R_{(\xi|\zeta)}(x) &= R_\xi(x) \wedge \neg R_\zeta(x)
    \end{align}
  \end{subequations}
  And then formally
  \begin{equation}
    I(\xi)=\sum_{\substack{x\in\D_\A\\R_\xi(x)}}I(x).
  \end{equation}
This agrees with both the intuition about these regions and the ways we've defined $I$ for some types of dependencies in the body of the paper.

For a more systematic development of the correspondence between information theory and set theory that associates generalized dependencies with elements of a $\sigma$-field see~\cite{sigmafield}.

\section{}
\label{sanitycheck}
In this Appendix we prove a few results that show the reasonableness of the definition of $C_{\A\'|\A}(k)$.
\begin{theorem}
  \label{thm:samesys}
  $C_{\A|\A}(k) = 0$
\end{theorem}
\begin{proof}
  \begin{subequations}
    \begin{align}
      C_{\A|\A}(k) &= \sum_{\substack{x\in\D_\A\\ s(x)\ge k}}I(x|x)\\
                   &= \sum_{\oldsubstack{x\in\D_\A\\s(x)\ge k}} \sum_{\oldsubstack{z\in\D_{\A\cup\A}\\z\le x\\z\not\le x}}I(z)\\
                   &= \sum_{\oldsubstack{x\in\D_\A\\s(x)\ge k}}\sum_{z\in\0}I(z)\\
                   &= 0
    \end{align}
  \end{subequations}
\end{proof}

The next few proofs involve the idea of independent systems. This is defined by the independence of the information contained in them, $H(A\cup B) = H(A) + H(B)$. In the following lemma we show that this means that the two systems have no shared information between them.

\begin{lemma}
  \label{lemma:indepsys}
  If $H(A\cup B) = H(A) + H(B)$ and $x\in\D_{\A\cup\B}$ has $\Inc(x)\cap A\ne\0$ and $\Inc(x)\cap B\ne\0$ then $I(x)=0$.
\end{lemma}
\begin{proof}
  Our proof rests in calculating $H(A\cup B)$ in two different ways. First, we can calculate it:
  \begin{subequations}
    \begin{align}
      H(A\cup B) &= \sum_{\substack{x\in\D_{\A\cup\B}\\\Inc(x)\cap(A\cup B)\ne\0}}I(x)\\
                 &= \sum_{\substack{x\in\D_{\A\cup\B}\\\Inc(x)\ins A}}I(x) 
      + \sum_{\substack{x\in\D_{\A\cup\B}\\\Inc(x)\ins B}}I(x) 
      + \sum_{\substack{x\in\D_{\A\cup\B}\\\Inc(x)\cap A\ne\0\\\Inc(x)\cap B\ne\0}}I(x).\\
      \intertext{On the other hand we can calculate it:}
      H(A\cup B) &= H(A) + H(B)\\
                 &= \sum_{\substack{x\in\D_{\A\cup\B}\\\Inc(x)\cap A\ne\0}}I(x) 
      + \sum_{\substack{x\in\D_{\A\cup\B}\\\Inc(x)\cap B\ne\0}}I(x)\\
      &= \sum_{\substack{x\in\D_{\A\cup\B}\\\Inc(x)\ins A}}I(x) 
      + \sum_{\substack{x\in\D_{\A\cup\B}\\\Inc(x)\cap A\ne\0\\\Inc(x)\cap B\ne\0}}I(x)
      + \sum_{\substack{x\in\D_{\A\cup\B}\\\Inc(x)\ins B}}I(x)
      + \sum_{\substack{x\in\D_{\A\cup\B}\\\Inc(x)\cap A\ne\0\\\Inc(x)\cap B\ne\0}}I(x)\\
      &= \sum_{\substack{x\in\D_{\A\cup\B}\\\Inc(x)\ins A}}I(x) 
      + \sum_{\substack{x\in\D_{\A\cup\B}\\\Inc(x)\ins B}}I(x)
      + 2\sum_{\substack{x\in\D_{\A\cup\B}\\\Inc(x)\cap A\ne\0\\\Inc(x)\cap B\ne\0}}I(x).\\
      \intertext{Subtracting the two equations, we get}
      0 &= \sum_{\substack{x\in\D_{\A\cup\B}\\\Inc(x)\cap A\ne\0\\\Inc(x)\cap B\ne\0}}I(x).
    \end{align}
  \end{subequations}
  Since this is a sum of non-negative terms, they must all be $0$, which was what we wanted.
\end{proof}

\begin{theorem}
  \label{thm:newsys}
  Given a system $\A$ and a new system $\B$ with no relation to $\A$, formally stated as $H(A\cup B) = H(A) + H(B)$, we have $C_{\B|\A}(k) = C_\B(k)$.
\end{theorem}
\begin{proof}
  \begin{subequations}
    \begin{align}
      I(x\'|x) &= \sum_{\substack{y\in\D_\A\\y\le x\'\\y\not\le x}}I(y) \\
               &= \sum_{\substack{y\in\D_\A\\y\le x\'}}I(y) - \sum_{\substack{y\in\D_\A\\y\le x\'\\y\le x}}I(y)\\
      \intertext{Since in the second sum $\Inc(y)\supseteq\Inc(x\')$ and $\Inc(y)\supseteq\Inc(x)$, we have that $\Inc(y)\cap A\ne\0$, and $\Inc(y)\cap B\ne\0$. By lemma~\ref{lemma:indepsys} $I(y)=0$.}
      &= I(x\') - \sum_{\substack{y\in\D_\A\\y\le x\'\\y\le x}}0\\
      &= I(x\')
    \end{align}
  \end{subequations}
\end{proof}

\begin{theorem}
  \label{thm:partsys}
  If $\A\'$ and $\B\'$ are independent systems, that is $H(A\'\cup B\')$ = $H(A\') + H(B\')$, then $C_{\A\'\cup\B\'|\A\cup\B}(k) = C_{\A\'|\A}(k) + C_{\B\'|\B}(k)$.
\end{theorem}
\begin{proof}
  This is essentially a direct result of lemma~\ref{lemma:indepsys}.
  \begin{subequations}
    \begin{align}
      C_{\A\'\cup\B\'|\A\cup\B}(k) &= \sum_{\substack{x\in\D_{\A\cup\B}\\s(x)\ge k}}I(x\'|x)\\
                                   &= \sum_{\substack{x\in\D_{\A\cup\B}\\s(x)\ge k\\\Inc(x)\ins A}}I(x\'|x)
      + \sum_{\substack{x\in\D_{\A\cup\B}\\s(x)\ge k\\\Inc(x)\ins B}}I(x\'|x)
      + \sum_{\substack{x\in\D_{\A\cup\B}\\s(x)\ge k\\\Inc(x)\cap B\ne\0\\\Inc(x)\cap A\ne\0}}I(x\'|x)\\
      &= \sum_{\oldsubstack{x\in\D_{\A}\\s(x)\ge k}}
      \sum_{\oldsubstack{y\in\D_{\A\cup\B}\\y\le x}}I(y\'|y)
      + \sum_{\oldsubstack{x\in\D_{\B}\\s(x)\ge k}}
      \sum_{\oldsubstack{y\in\D_{\A\cup\B}\\y\le x}}I(y\'|y)
      + \sum_{\substack{x\in\D_{\A\cup\B}\\s(x)\ge k\\\Inc(x)\cap B\ne\0\\\Inc(x)\cap A\ne\0}}I(x\'|x)\\
      &= \sum_{\substack{x\in\D_\A\\s(x)\ge k}}I(x\'|x)
      + \sum_{\substack{x\in\D_\B\\s(x)\ge k}}I(x\'|x)
      + \sum_{\substack{x\in\D_{\A\cup\B}\\s(x)\ge k\\\Inc(x)\cap B\ne\0\\\Inc(x)\cap A\ne\0}}0\\
      &= C_{\A\'|\A}(k) + C_{\B\'|\B}(k)
    \end{align}
  \end{subequations}
\end{proof}

We can now reiterate our result from Section \ref{simple} that in the case where a subsystem of the system deviates completely, we can treat this as the subsystem becoming part of the environment. In our new theory, this means that the complexity of the noise is that of the deviating subsystem.
\begin{corollary}
  Given a system $\A$ with an independent subsystem $\B$ that deviates completely from it's defined behavior, and now behaves as $\B\'$, then $C_{\A\'|\A}(k) = C_{\B\'}(k)$
\end{corollary}
\begin{proof}
  Theorem~\ref{thm:partsys} gives us that $C_{A\'|\A}(k) = C_{\A\'\setminus\B\'|\A\setminus\B}(k) + C_{\B\'|\B}(k)$. Theorem~\ref{thm:samesys} gives us that the first term is $0$ (since that part of the system does not change) and Theorem~\ref{thm:newsys} gives us that $C_{\B\'|\B}(k) = C_{\B\'}(k)$, since $B\'$ has completely new behavior from $\B$.
\end{proof}

\section{}
\label{apx:noisehascx}
In this Appendix we prove that when complexity arises in a system at a certain scale, the complexity profile of the change has that complexity as well.
\begin{theorem}
  \label{noisehascx}
  Given some scale $\kappa$, If $C_{\A}(\kappa) = 0$, then $C_{\A\'|\A}(\kappa) = C_{\A\'}(\kappa)$. In particular, if $C_{\A\'}(\kappa) > 0$, then $C_{\A\'|\A}(\kappa) > 0$.
\end{theorem}
\begin{proof}
  $C_{\A}(\kappa) = \sum_{\oldsubstack{x\in\D_\A\\s(x)\ge \kappa}}I(x) = 0$, and by assumption for all $x$, $I(x)\ge 0$ so for all $x\in\D_\A$ with $s(x)\ge \kappa$, and we have $I(x)=0$.
  Now we can write
  \begin{subequations}
    \begin{align}
      C_{\A\'|\A}(\kappa) &= \sum_{\substack{x\in\D_\A\\s(x)\ge \kappa}}I(x\'|x)\\
                          &= \sum_{\oldsubstack{x\in\D_\A\\s(x)\ge \kappa}}\sum_{\oldsubstack{z\in\D_{\A\cup\A}\\z\le x\'\\z\not\le x}}I(z)\\
                          &= \sum_{\oldsubstack{x\in\D_\A\\s(x)\ge \kappa}}\left(\sum_{\oldsubstack{z\in\D_{\A\cup\A\'}\\z\le x\'}}I(z) - \sum_{\substack{z\in\D_{\A\cup\A\'}\\z\le x\'\\z\le x}}I(z)\right).\\
      \intertext{The second sum in the parentheses (call it $\Gamma$) is a sum over nonnegative values, so it must be nonnegative itself ($0\le\Gamma$). But it also has a subset of the terms in $I(x)$ (with $s(x)\ge\kappa$), so we have $\Gamma\le I(x)=0$, so $\Gamma=0$. The first sum in the parentheses is exactly the definition of $I_{A\cup\A\'}(x\')$. By lemma~\ref{lemma:Iwelldefined}, $I_{\A\cup\A\'}(x\')=I_{\A\'}(x\')=I(x\')$, so we have} 
      &= \sum_{\substack{x\in\D_\A\\s(x)\ge \kappa}}I(x\')\\
      &= \sum_{\substack{x\'\in\D_{\A\'}\\s(x\')\ge \kappa}}I(x\')\\
      &= C_{\A\'}(\kappa)
    \end{align}
  \end{subequations}
\end{proof}

\end{document}